\documentclass{llncs}
\usepackage{amsmath}
\usepackage{amsfonts}
\usepackage{amssymb}
\usepackage{graphicx}
\usepackage{psfrag}
\usepackage{tikz}
\usetikzlibrary{intersections, calc}
\usepackage{algorithm}
\usepackage{algpseudocode}
\usepackage{grffile}
\usepackage{placeins}
\usepackage{color}

\begin{document}
\frontmatter          
\pagestyle{headings}  
\mainmatter              
\title{Edge Elimination in TSP Instances}
\titlerunning{Edge Elimination in TSP Instances}  
%
\author{Stefan Hougardy \and Rasmus~T.~Schroeder}
\authorrunning{ S.~Hougardy \and R.~T.~Schroeder} 
%
\institute{Research Institute for Discrete Mathematics,
University of Bonn\\
}

\maketitle              

\begin{abstract}
The Traveling Salesman Problem is one of the best studied NP-hard problems 
in combinatorial optimization. 
Powerful methods have been developed over the last 60 years to find 
optimum solutions to large TSP instances. The largest TSP instance so far that 
has been solved optimally has 85,900 vertices. Its solution required more than 136~years 
of total CPU time using the branch-and-cut based Concorde TSP code~\cite{ABCC2006}.

In this paper we present graph theoretic results that allow to prove that some edges of a TSP instance
cannot occur in any optimum TSP tour. Based on these results we propose a combinatorial algorithm
to identify such edges.
The runtime of the main part of our algorithm is $O(n^2 \log n)$ for
an $n$-vertex TSP instance. 
By combining our approach with the Concorde TSP solver we are able to solve a large TSPLIB instance more
than 11~times faster than Concorde alone.
\keywords{traveling salesman problem, exact algorithm}
\end{abstract}

\section{Introduction}

An instance of the Traveling Salesman Problem (TSP for short) consists of a 
complete graph on a vertex set $V$ together with a symmetric length 
function $l: V\times V\to \mathbb{R_+}$.
A \emph{tour} $T$ is a cycle that contains each vertex of the graph exactly once.
The length of a tour $T$ with edge set $E(T)$ is defined as $\sum_{e\in E(T)} l(e)$.
A tour $T$ for a TSP instance is called \emph{optimum} if no other tour for this 
instance has smaller length. Finding such an optimum TSP tour is a well known 
NP-hard problem~\cite{GarJoh1979}.

The Traveling Salesman Problem is one of the best studied problems in 
combinatorial optimization. Many exact and approximate algorithms have been 
developed over the last 60 years.
In this paper we present several theoretical results 
that allow us to eliminate edges from a
TSP instance that provably cannot be contained in any optimum TSP tour.
Based on these results we present a combinatorial algorithm that identifies such edges. 
As the runtime of our main algorithm is only $O(n^2 \log n)$ for an $n$-vertex instance, it
can be used as a preprocessing step to other TSP algorithms. 
On large instances our algorithm can speed up 
the runtime of existing exact TSP algorithms significantly.
It also can improve the performance of heuristic algorithms for the TSP. 
We present examples for both applications in Section~\ref{sec:results}.
For a good description of the state of the 
art in algorithms for the Traveling Salesman Problem see~\cite{ABCC2006}.

Our results are motivated by studying 2-dimensional
Euclidean TSP instances, i.e., instances where the vertices are
points in the Euclidean plane and the length of an edge is
the Euclidean distance between the two corresponding points.
However, most of our results hold for arbitrary symmetric TSP instances
that even do not need to be metric.

\subsubsection*{Our Contribution.}
We present several results that allow to prove that certain edges in a TSP instance cannot 
belong to any optimum TSP tour. Our Main Edge Elimination Theorem that we prove in 
Section~\ref{sec:main} turns out to be quite powerful. It allows 
to reduce the $n(n-1)/2$ edges of an $n$-vertex TSP instance to about $30 n$ edges  
or less for the TSPLIB~\cite{Rei1995} instances. 
Some additional methods for eliminating edges are presented in 
Section~\ref{sec:closepoint}. We combine these with a backtrack search which
we present in Section~\ref{sec:backtracking}.
This will allow us to reduce the number of edges in
the TSPLIB instances to about $5n$ edges.  

Unfortunately, our Main Edge Elimination Theorem does not directly lead to an efficient
algorithm for eliminating edges in a TSP instance. Thus, a second major contribution of this
paper is contained in Section~\ref{sec:validating} where we prove 
that a weaker form of our Main Edge Elimination Theorem can be checked 
in constant time per edge. 
This allows to apply our methods to very large TSP instances
containing 100,000 or more vertices. The total runtime on such a large instance is less than three days 
on a single processor.
Our algorithm can be run in parallel on all edges resulting in a runtime of less than a 
minute if sufficiently many processors are available.

Section~\ref{sec:results} contains the results of our algorithm on TSPLIB~\cite{Rei1995} 
instances as well as on a 100,000 vertex instance. Here we also show how 
our approach can speed up finding optimum solutions to large TSP instances significantly.
The TSP solver Concorde~\cite{ABCC2006} is the fastest available algorithm 
to solve large TSP instances optimally.
Concorde needs more than 199 CPU days for the TSPLIB instance d2103.
After running our edge elimination algorithm for 2~CPU days the runtime of Concorde decreases 
to slightly more than 16 CPU days. The total speed up we obtain is more than a factor of~11.

We also report two other successful applications of the edge elimination approach in Section~\ref{sec:results}.

\section{Notation and Preliminaries}
\label{sec:notation}

To avoid some degenerate cases we assume in this paper that a TSP instance 
contains at least four vertices. 
Edges that do not belong to any 
optimum TSP tour will be called \emph{useless}. We extend the definition of
a TSP instance to instances $(V,E)$ together with a length function $l:E\to\mathbb{R}_+$ 
where $E$ contains all edges of the
complete graph on $V$ except some useless edges.  This implies that 
\emph{all} optimum TSP tours on the complete graph on $V$ are contained in the graph 
$(V,E)$. Therefore, to find an optimum TSP tour it suffices to consider edges from $E$.
However, to prove that some edge in $E$ is useless it might be useful to also look at edges
that are not in $E$. 
Our edge elimination algorithm will start with some TSP instance $(V,E)$ and 
return an instance $(V,E')$ such that $E'$ is a subset of $E$ and contains all
optimum TSP tours.

As most of our results are inspired by studying Euclidean instances we will 
often call a vertex in a TSP instance a \emph{point}.
If $x$ and $y$ are two vertices in a TSP instance we will denote the edge 
$\{x,y\}$ by $xy$ simply.

A TSP instance is called \emph{metric} if $l(xy) \le l(xz) +  l(zy)$ for all
vertices $x,y,z\in V$. 
Our results do not require that the TSP instance is metric. However, we can 
improve some of our results if we assume some metric properties of the instance.
Euclidean TSP instances are of course metric, but  
because of problems with floating point accuracy such instances are usually not 
considered in practice. 
The well known 
TSPLIB~\cite{Rei1995} instances for example use the discretized Euclidean 
distances EUC\_2D and CEIL\_2D. In the first, the Euclidean distance is rounded 
to the nearest integer while in the second it is rounded up
to the next integer.
Note that the EUC\_2D distance function is not metric and that an optimum TSP 
tour for such an instance may contain two crossing edges.

Currently, the most successful heuristic TSP algorithms~\cite{Hel2009}
are based on the concept of \emph{$k$-opt moves}.
Given a TSP tour a $k$-opt move makes local changes to the tour by replacing $k$ 
edges of the tour by $k$ other edges. For a $k$-opt move we require, that after 
the replacement of the $k$ edges the new subgraph is 2-regular.
If the new subgraph is again a tour we call the 
$k$-opt move \emph{valid}. If a tour $T$ allows a valid $k$-opt move resulting 
in a shorter tour, then $T$ cannot be an optimum tour. This simple observation 
is the core of our algorithm for proving the existence of useless edges. \medskip

Let $pq$ and $xy$ be two edges in a TSP instance. We call $pq$ and $xy$ 
\emph{compatible}, denoted by $pq\sim xy$, if
\begin{equation} \label{eqn:compatible}
\max \left(l(px)+l(qy), l(py)+l(qx)\right) ~\ge~ l(pq) + l(xy)~.
\end{equation}
Otherwise $pq$ and $xy$ are called \emph{incompatible}. Note that two edges that 
have at least one vertex in common are always compatible.

\begin{lemma} \label{lemma:compatibleedges}
Any two edges in an optimum TSP tour are compatible.
\end{lemma}

\begin{proof}
Assume $pq$ and $xy$ are two incompatible edges in an optimum TSP tour $T$.
By (\ref{eqn:compatible}) we have 
$l(px)+l(qy) ~<~ l(pq) + l(xy)$ and $l(py)+l(qx) ~<~ l(pq) + l(xy)$.
Thus $T$ can be improved by a $2$-opt move, that replaces edges $pq$ and $xy$ by
either $px$ and $qy$ or by $py$ and $qx$. One of these two $2$-opt moves must be 
valid. This contradicts the assumption that $T$ is an optimum TSP tour.
\qed
\end{proof}

For $k>2$ we call a set of $k$
edges \emph{$k$-incompatible}, if they cannot belong to the same optimum
TSP tour.

\section{The Main Edge Elimination Theorem}
\label{sec:main}

To be able to formulate our Main Edge 
Elimination Theorem, we need to introduce the concept of \emph{potential 
points} first.
\label{sec:potentialpoints}
Let $(V,E)$ be a TSP instance and $pq\in E$. For $r\in V\setminus \{p,q\}$ 
define 
\begin{equation}
R:= \{x\in V~|~rx \in E\land pq\sim rx\}~.
\end{equation}
Let $R_1,\ R_2\subset V$ and $R\subset R_1\cup R_2$.
We call $r$ \emph{potential} with respect to $pq$ and $R_1$ and $R_2$, if for 
every optimum tour containing $pq$, the two neighbors of $r$ cannot both lie in 
$R_1$ respectively $R_2$. We say that such a covering \emph{certifies} the 
potentiality of $r$. Note that $R_1\cap R_2$ needs not to be empty. However if 
$r$ is potential, in any optimum tour containing $pq$, $r$ cannot be connected 
with $R_1\cap R_2$. In Section~\ref{sec:certification} we will develop efficient 
methods for certifying potential points.

\begin{theorem}[Main Edge Elimination Theorem]\label{thm:main}%
Let $(V,E)$ be a TSP instance and $pq\in E$. Let $r$ and $s$ be two different 
potential points with respect to $pq$ with covering $R_1$ and $R_2$ 
respectively $S_1$ and $S_2$. Let $r\not\in S_1\cup S_2$ and $s\not\in R_1\cup
R_2$. If
\begin{equation}
\label{eqn:main1}
l(pq) - l(rs) + \min_{z\in S_1} \left\{ l(sz) - l(pz)\right\} + \min_{y\in R_2}
\left\{ l(ry) - l(qy)\right\} ~ > ~ 0
\end{equation}
and
\begin{equation}
\label{eqn:main2}
l(pq) - l(rs) + \min_{x\in R_1} \left\{ l(rx) - l(px)\right\} + \min_{w\in S_2}
\left\{ l(sw) - l(qw)\right\} ~ > ~ 0~,
\end{equation}
then the edge $pq$ is useless.
\end{theorem}

\begin{proof} 
Assume that the edge $pq$ is contained in an optimum TSP tour $T$. 
Let $rx, ry, sz, sw\in T$ be the incident edges of $r$ and $s$. 
We may assume that the vertices $x,y,z$, and $w$ are labeled in such a
way that $x\in R_1$, $y\in R_2$, $z\in S_1$, and $w\in S_2$.
As $r$ and $s$ are potential, we have $r,s\not\in \{p,q\}$.
By assumption we have $rs\not\in T$, making the four edges $rx, ry, sz$ and $sw$ 
distinct.

\begin{figure}[ht]
\psfrag{p}[b][]{\raisebox{2mm}{$p$}}
\psfrag{q}[b][]{\raisebox{2mm}{$q$}}
\psfrag{r}[rt][]{\raisebox{-2mm}{$r$~}}
\psfrag{s}[rb][]{\raisebox{-2mm}{$s$~}}
\psfrag{a}[r][]{$x$~}
\psfrag{b}[l][]{~$y$}
\psfrag{c}[r][]{$z$~}
\psfrag{d}[l][]{~$w$}
\begin{tabular}{l@{\hspace*{13.8mm}}l@{\hspace*{13.8mm}}l}
\includegraphics[width=3cm]{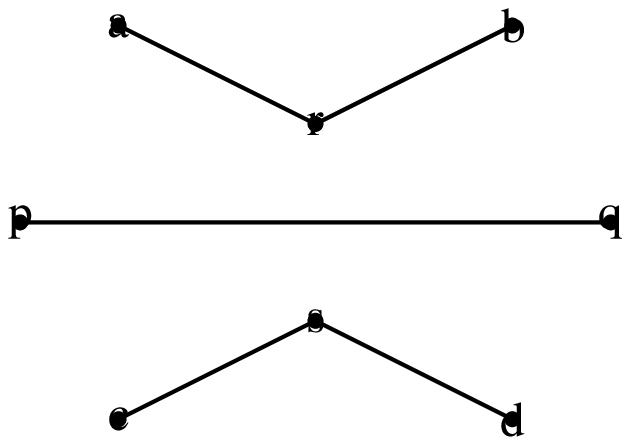} &
\includegraphics[width=3cm]{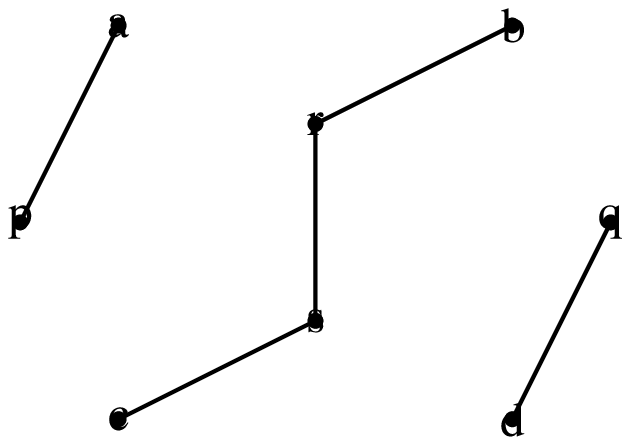} &
\includegraphics[width=3cm]{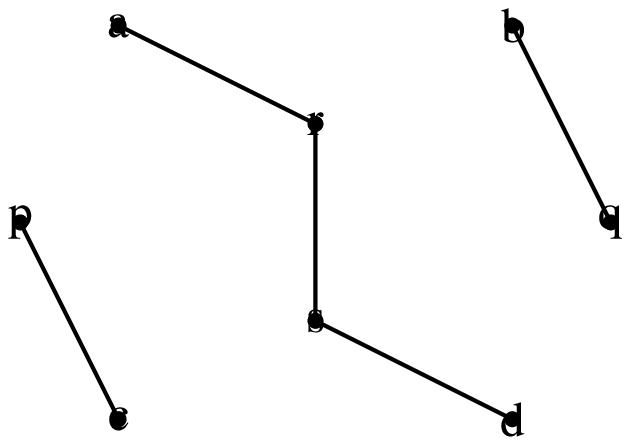} \\[2mm]
The situation of \rlap{Theorem~\ref{thm:main}.} &
One possible 3-opt. &
Another possible 3-opt. \\
\end{tabular}
\caption{Two possible 3-opt moves that imply that the edge $pq$ is useless.}
\label{fig:MainEdgeElimination}
\end{figure}

Now there exist two possible $3$-opt moves as shown in 
Figure~\ref{fig:MainEdgeElimination}. The first is to replace $pq$, $rx$, and 
$sw$ with $px$, $rs$, and $qw$. The second is to replace $pq$, $ry$, and $sz$ 
with $pz$, $rs$, and $qy$. It is easy to verify that for every tour containing 
the egdes $pq, rx, ry, sz$ and $sw$, one of these two 3-opt moves must be valid.

The two $3$-opt moves are decreasing the length of the tour $T$ by
\begin{align}
&l(pq) - l(rs) + l(rx) - l(px) + l(sw) - l(qw)~\text{respectively}\\
&l(pq) - l(rs) + l(ry) - l(py) + l(sz) - l(qz)~.
\end{align}
By inequalities~(\ref{eqn:main1}) and (\ref{eqn:main2}), both terms are strictly
positive. Since one of these $3$-opt moves is valid, this yields a tour shorter
than $T$, contradicting the optimality of $T$.
\qed
\end{proof}

\section{The Close Point Elimination Theorems}
\label{sec:closepoint}

The Main Edge Elimination Theorem will be our primary tool to prove that an edge in a TSP 
instance is useless. As soon as many edges of a TSP instance 
are known to be useless other methods can be applied. In this section
we present our so called \emph{Close Point Elimination}. When applied to the complete graph
of a TSP instance it will eliminate almost no edge. However,
in combination with the Main Edge Elimination Theorem it will allow to identify additional 
useless edges.

\begin{theorem}[Close Point Elimination Theorem]\label{thm:ClosePoint}%
Let $(V,E)$ be a TSP instance and $pq\in E$. Let $r\in V\setminus\{p,q\}$ and 
define $R:= \left\{ x\in V~|~rx\in E \land pq \sim rx\right\}$.
If for all $x,y\in R$ with $\{x,y\} \not =  \{p,q\}$ we have
\begin{equation}
\label{eqn:closepointelimination}
l(xy) + l(pr) + l(qr) ~<~ l(pq) + l(xr) + l(yr)
\end{equation}
then the edge $pq$ is useless.
\end{theorem}

\begin{proof}
Assume that an optimum tour $T$ contains the edge $pq$. Let $xr$ and $ry$ be the 
two edges in $T$ that are incident with $r$. Then $\{x,y\} \not =  \{p,q\}$ 
and $x$ and $y$ must be in $R$. By assumption inequality~(\ref{eqn:closepointelimination})
holds. Then we can replace the edges 
$pq$, $xr$, and $yr$ with $xy$, $pr$, and $qr$ and obtain a tour that is shorter 
than $T$. Note that this also holds if one of $x$ and $y$ equals one 
of $p$ and $q$
. This contradicts the optimality of the tour $T$.
\qed
\end{proof}

For the degenerate case with $x=p$ we show a stronger result by
making use of the notion of \emph{metric excess}.
It will allow us to short cut a eulerian subgraph in an instance that needs not to be metric.
The \emph{metric excess} $m_{pq} (z)$ of a vertex $z$ 
with respect to an edge $pq$ is defined as
\begin{eqnarray*}
\min_{x,y\in N(z)\setminus\{p,q\}} 
& ~~\max~ \{ & l(xz) + l(zp) - l(xp), l(yz) + l(zp) - l(yp), \\
&         & l(xz) + l(zq)- l(xq), l(yz) + l(zq) - l(yq) ~\}~.
\end{eqnarray*}

\begin{theorem}[Strong Close Point Elimination Theorem]
\label{lemma:3incompatible}
Let $pq$, $pr$ and $rx$ be three edges of a TSP instance $(V,E)$. Let 
$z\in V\setminus\{p,q,r,x\}$. If
\begin{equation}
l(xq) + l(rz) + l(zp) - m_{pr}(z) ~<~ l(pq) + l(rx)~,
\end{equation}
then the edges $pq$, $pr$ and $rx$ are $3$-incompatible.
\end{theorem}

\begin{proof}
Assume that an optimum tour $T$ contains the edges $pq$, $pr$ and $rx$. 
We show that there exists a 3-opt move yielding a tour shorter than $T$.
Delete the edges $pq$ and $rx$ and insert the edges $qx$, $pz$ and $rz$.
Note that this edge set is eulerian but not a TSP tour as vertex $z$ has degree 
four. But as $ l(xq) + l(rz) + l(zp) - m_{pr}(z) ~<~ l(pq) + l(rx) $
a short cut is possible that yields a tour shorter than $T$.
This contradicts the optimality of the tour $T$.
\qed
\end{proof}

\section{Certifying potential points}
\label{sec:certification}

Starting with a complete graph, it can be very inefficient, to exclude
edges using the criteria discussed so far. We will show how to adopt the 
Main Edge Elimination Theorem such that most edges can be excluded efficiently.

The aim of this section is to show that one can prove in constant time that a 
point $r$ is potential with respect to an edge $pq$. 
To be able to do so, we will assume in the following 
that the distance function $l$ is the EUC\_2D function. Similar results hold for 
any other distance function that is close to the Euclidean distance.
By $|pq|$ we denote the Euclidean distance of $p$ and $q$. Note that
\begin{equation}
l(pq) - \frac{1}{2} \leq |pq| \leq l(pq) + \frac{1}{2}~. \label{eq:euc2d}
\end{equation}

\label{sec:partition}

We now want to find a covering for a point $r$ and an edge $pq$. We will 
therefore take a closer look at the set 
$R:= \{x\in V~|~rx \in E \land pq\sim rx\}$
which was already defined in Section~\ref{sec:potentialpoints}.
Observe that in a metric space the compatibility of the edges $pq$ and $rs$ 
implies that the edges $pq$ and $rt$ are also compatible for all $t\in rs$. 
Since we use EUC\_2D lengths, this only holds after adding some constants. 
For each point $r$ choose $\delta_r$ s.t. no vertex apart from $r$ lies in
the interior of the circle around $r$ with radius $\delta_r$. One can for
example use
\begin{equation}
\delta_r := \frac{1}{2} + \max\{d\in \mathbb{Z}_+ 
~|~ \forall s\in V\setminus \{r\} ~ l(rs) > d\}~.
\end{equation}
For an edge $pq$ and a point $r\in V\setminus\{p,q\}$ define the two lengths
\begin{eqnarray}
l_p := \delta_r + l(pq) - l(qr) - 1~~~\text{and}~~~\label{def:lp} 
l_q := \delta_r + l(pq) - l(pr) - 1~.\label{def:lq}
\end{eqnarray}
For each point $s\in V\setminus \{r\}$ define $s_r \in rs$ such that $|rs_r| = \delta_r$.

\begin{lemma}
\label{thm:conelemma}
Let $(V,E)$ be a TSP instance, $pq\in E$, $r\in V\setminus \{p,q\}$ and 
$s\in V\setminus \{r\}$. If 
$|ps_r| < l_p$ and $|qs_r| < l_q$
then the edges $pq$ and $rs$ are incompatible.
\end{lemma}

\begin{proof}
We show that both $2$-opt moves involving the edges $pq$ and $rs$ have shorter
length. Using~(\ref{eq:euc2d}), (\ref{def:lp}), and the triangle inequality we get:
\begin{eqnarray*}
l(ps) + l(qr) &\leq& |ps_r| + |ss_r| + \frac{1}{2} + l(qr) 
{~<~} l_p + |ss_r| + \frac{1}{2} + l(qr)\\
&=& \delta_r + l(pq) - l(qr) - \frac{1}{2} + |ss_r| + l(qr)
~\leq~ l(pq) + l(rs)
\end{eqnarray*}
\noindent
$l(qs) + l(pr)< l(pq) + l(rs)$ is proven analogously. Hence $pq$ and $rs$ are
incompatible.
\qed
\end{proof}

\begin{figure}[ht]
\begin{center}
\psfrag{p}[b][]{\raisebox{2mm}{$p$}}
\psfrag{q}[b][]{\raisebox{2mm}{$q$}}
\psfrag{Rp}[t][]{\raisebox{0.5mm}{~$R_p$}}
\psfrag{Rq}[t][]{\raisebox{2.5mm}{~~$R_q$}}
\psfrag{x}[rt][]{\raisebox{-2mm}{$r$~}}
\includegraphics[height=5cm]{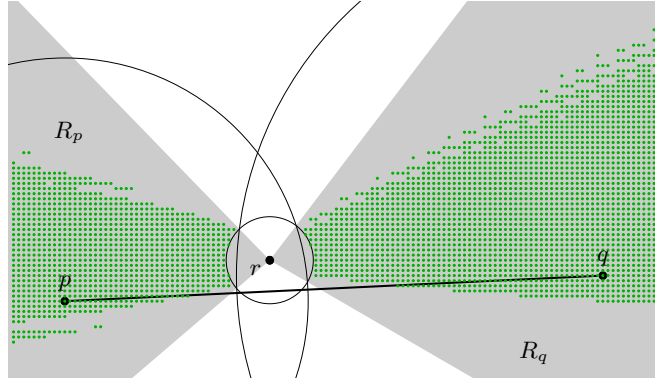}
\caption{The small dots indicate all possible endpoints for the edge $rs$ such 
that $rs \sim pq$. By Lemma~\ref{thm:conelemma}
all these vertices must be contained in the gray cones $R_p$ and $R_q$.}
\label{fig:RpRqCones}
\end{center}
\end{figure}

Figure~\ref{fig:RpRqCones} illustrates Lemma~\ref{thm:conelemma}.
It shows two cones $R_p$ and $R_q$ for which $R\subset R_p\cup R_q$. 
The cones are defined as
$R_p := \{t~|~|qt_r|\geq l_q\}$ and
$R_q := \{t~|~|pt_r|\geq l_p\}$.
We do not need that the cones $R_p$ and $R_q$ are disjoint.
However we need that they actually exist, i.e. the circles around $p$ and $q$ 
have to intersect the circle around $r$.

\begin{lemma}
\label{thm:circleintersection}
Let $(V,E)$ be a TSP instance, $pq\in E$ and $r\in V\setminus \{p,q\}$. If
\begin{equation}
\label{eqn:basic}
l_p + l_q \geq l(pq) - \frac{1}{2}~,
\end{equation}
then the circle with center $r$ and radius $\delta_r$ intersects both circles 
with centers $p$ and $q$ and radii $l_p$ respectively $l_q$.
\end{lemma}

\begin{proof}
It suffices to show 
$|ir| - \delta_r \leq l_i \leq |ir| + \delta_r$ for $i\in \{p,q\}$.
\begin{eqnarray*}
|pr| - \delta_r &\leq& l(pr) + \frac{1}{2} - \delta_r
\overset{(\ref{def:lq})}{=} l(pq) - \frac{1}{2} - l_q
\overset{(\ref{eqn:basic})}{\leq} l_p \quad \\
&\overset{(\ref{def:lp})}{\leq}& 
(|pr| + |qr| + \frac{1}{2}) - (|qr| - \frac{1}{2}) + \delta_r - 1
= |pr| + \delta_r
\end{eqnarray*}
Analogously for $i=q$.
\qed
\end{proof}

A naive way to prove that $R_p$ and $R_q$ certify the potentiality of $r$ 
is to look at all possible 3-opt moves.

\begin{lemma}\label{lemma:quadraticcheck}
Let $(V,E)$ be a TSP instance and $pq\in E$. For $r\in V\setminus \{p,q\}$ let 
$R\subset R_p\cup R_q$ with $R = \{x\in V~|~rx\in E\land pq\sim rx\}$.
If for $i\in \{p,q\}$:
\begin{equation}
\label{eqn:quadraticcertification}
l(pq) + l(rx) + l(ry) ~>~ l(pr) + l(rq) + l(xy) ~~~~\text{for all}~x,y\in R_i,
\end{equation}
then $R_p$ and $R_q$ certify the potentiality of $r$.
\end{lemma}

\begin{proof}
Assume that $pq$ is contained in an optimum tour $T$ and $rx, ry\in T$ with
$x, y\in R_i$ for $i\in \{p,q\}$. Then replacing the edges $pq$, 
$rx$ and $ry$ by the edges $pr$, $rq$ and $xy$ is a valid 3-opt move.
By inequality~(\ref{eqn:quadraticcertification}) this 3-opt move yields a 
shorter tour, contradicting the optimality of $T$.
\qed
\end{proof}

Lemma~\ref{lemma:quadraticcheck} yields a method to check in $O(n^2)$ time that a vertex $r$ is potential.
We now want to show that the potentiality of a vertex $r$ can be certified in 
constant time. We assume that the edge $pq$ is part of an optimum tour $T$. 
Furthermore we consider the covering $R_p$ and $R_q$ as described above.
The angles of the cones $R_p$ and $R_q$ are denoted by $\alpha_p$ respectively 
$\alpha_q$. The certification follows from a simple argument.

\begin{lemma}
\label{lemma:gammagreateralpha}
Assume that $pq$ is contained in an optimum TSP tour $T$, $r\in V\setminus 
\{p,q\}$ and the angle $\gamma$ between the two edges incident with $r$ in $T$ 
satisfies
\begin{equation}
\label{eqn:gammagreateralpha}
\gamma > \text{max}\{\alpha_p,\alpha_q\}~.
\end{equation}
Then the neighbors of $r$ in $T$ cannot both lie in $R_p$ respectively $R_q$.
\end{lemma}

\begin{proof}
W.l.o.g. assume both neighbors of $r$ in $T$ lie in $R_p$. This immediatly
implies $\gamma\leq \alpha_p$, contradicting~(\ref{eqn:gammagreateralpha}).
\qed
\end{proof}

It now suffices to show that inequality~(\ref{eqn:gammagreateralpha}) holds for 
every optimum tour containing $pq$.
This can be checked using the following statement.

\begin{lemma}
\label{lemma:gamma}
Let $(V,E)$ be a TSP instance and $T$ an optimum tour. Let $pq\in T$ and 
$r\in V\setminus \{p,q\}$. Assume that inequality~(\ref{eqn:basic}) holds. 
Define the angle $\gamma_r$ as
\begin{equation}
\gamma_r := \arccos \left( 1 - \frac{\left(l_p + l_q - l(pq) + \frac{1}{2} 
\right)^2}{2\delta_r^2} \right)~.
\end{equation}
Then the angle $\gamma$ between the two edges of $T$ incident with 
vertex $r$ satisfies
\begin{equation}
\label{eqn:gammaestimation}
\gamma \geq \gamma_r~.
\end{equation}
\end{lemma}

\begin{proof}
Let $rx, ry\in T$ be the two incident edges of $r$. Let $\mu :=|x_ry_r|$ .
The cosine formula yields the following equation:
\begin{equation}
\label{eq:cos-thm}
\mu^2 = 2\delta_r^2 - 2\delta_r^2\cos\gamma
\end{equation}
As $T$ is an optimum tour, there is no valid 3-opt move which yields a shorter
tour. Hence we get:
\begin{align*}
& & l(pq) + l(rx) + l(ry) &\leq l(pr) + l(qr) + l(xy)\\
&\Rightarrow & l_p + l_q + |x_rx| + |y_ry| -l(pq) + 1 &\leq l(xy)\\
&\Rightarrow & l_p + l_q -l(pq) + \frac{1}{2} &\leq \mu\\
&\overset{(\ref{eqn:basic})}{\Rightarrow} & \left(l_p + l_q -l(pq) + \frac{1}{2} 
\right)^2 &\leq \mu^2 ~=~ 
2\delta_r^2 - 
2\delta_r^2\cos\gamma\\
&\Rightarrow & \cos\gamma &\leq 1 - \frac{\left(l_p + l_q -l(pq) + \frac{1}{2} 
\right)^2}{2\delta_r^2}\\
&\Rightarrow & \gamma &\geq \gamma_r~. \\[-8mm]
\end{align*}
\qed
\end{proof}

From Lemma~\ref{lemma:gammagreateralpha} and Lemma~\ref{lemma:gamma} we  
immediately get the following result.

\begin{lemma}
Let $pq$ be an edge contained in some optimum TSP tour $T$ and $r\in V\setminus 
\{p,q\}$. Assume that inequality~(\ref{eqn:basic}) holds. If
\begin{equation}
\label{eqn:allgammagreateralpha}
\gamma_r > \text{max}\{\alpha_p, \alpha_q\}~,
\end{equation}
then the sets $R_p$ and $R_q$ certify the potentiality of $r$.
\end{lemma}

It is easy to see that the angles $\alpha_p$ and $\alpha_q$ can be computed in 
constant time. Details are given in the appendix.\label{sec:alphas}

The results so far provide a way to prove in constant time that a given vertex 
is potential.
Notice that not all potential points can be detected using this approach.
However, as we will see in Section~\ref{sec:results}, sufficiently many 
potential points can be detected using this method.
For simplifying notation we introduce the following concept:
Let $pq$ be an edge and $r\in V\setminus \{p,q\}$. The vertex $r$ is called
\emph{strongly potential} (with respect to $pq$), if the
conditions~(\ref{eqn:basic}) and (\ref{eqn:allgammagreateralpha}) hold.

Thus, checking whether a point $r$ is strongly potential can be done in constant 
time (assuming that the value $\delta_r$ is known, which can be computed in a 
preprocessing step for all vertices). Verifying the inequalities in the Main 
Edge Elimination Theorem still needs $O(n)$ time. 
\label{sec:validating}
The aim now is to show that this can be 
done in constant time by computing appropriate lower bounds for (\ref{eqn:main1}) and (\ref{eqn:main2}).

\begin{lemma}
\label{lemma:min-estimation}
Let $(V,E)$ be a TSP instance and $r$ strongly potential with respect to $pq$. 
Let $R_p$ and $R_q$ be the covering certifying $r$. Then
\begin{eqnarray}
\label{eqn:MinMax-p}
\underset{x \in R_p} {\text{min}}\{ l(rx) - l(px) \} &\geq& \delta_r - 1
- \text{max}\{|px_r|: x\in R_p\} \text{ and}\\
\label{eqn:MinMax-q}
\underset{y \in R_q} {\text{min}}\{ l(ry) - l(qy) \} &\geq& \delta_r - 1
- \text{max}\{|qy_r|: y\in R_q\}~.
\end{eqnarray}
\end{lemma}

\begin{proof}
Let $x\in R_p$. Then
\begin{eqnarray*}
l(rx)-l(px) &\geq& |rx|-|px|-1  
       ~\geq~ \delta_r + |x_rx|- (|px_r|+|x_rx|) - 1\\
       &\geq& \delta_r - 1 - \text{max}\{|px_r|: x\in R_p\}
\end{eqnarray*}
Similarly one can prove this for the set $R_q$.
\qed
\end{proof}

Let $C_r$ be the circle around $r$ with radius $\delta_r$.
Define the two arcs
\begin{eqnarray*}
B_p := \{x\in C_r~|~|qx| \geq l_q\}~~~\text{and}~~~ 
B_q := \{y\in C_r~|~|py| \geq l_p\}~.
\end{eqnarray*}
Further let $\tilde{p}$ and $\tilde{q}$ be the points on $C_r$ with greatest 
distance to $p$ respectively $q$.
Since $B_p$ and $B_q$ are connected, the maxima in the inequalities~(\ref{eqn:MinMax-p}) 
and~(\ref{eqn:MinMax-q}) can only be attained at 
$\tilde{p}$ respectively $\tilde{q}$, or at the endpoints of $B_p$ respectively
$B_q$.
We consider the case that
\begin{eqnarray}
|p\tilde{q}| \leq l_p \quad ~~\text{and}~~~ 
|q\tilde{p}| \leq l_q~. \label{eqn:pq-tilde}
\end{eqnarray}
This gives that
\begin{eqnarray}
\text{max}\{|px_r|: x\in R_p\} &\leq& \text{max}\{t~|~t\in B_p\} \quad 
\text{and}\label{eqn:MaxBp}\\
\text{max}\{|qy_r|: y\in R_q\} &\leq& \text{max}\{t~|~t\in B_q\}~,
\label{eqn:MaxBq}
\end{eqnarray}
where the right hand sides can easily be calculated in constant time. Details are given in the appendix.

\section{The Algorithm}
\label{sec:algorithm}

In this section we describe our algorithm that eliminates useless edges. 
It consists of three independent steps. Step~1 uses the results of 
Section~\ref{sec:main} and Section~\ref{sec:certification} and eliminates the majority of all 
edges. Step~2 applies the Main Edge Elimination in combination with the 
Close Point Elimination to eliminate additional edges. 
Finally in Step~3 we use a backtrack search of bounded depth to eliminate even more edges.

\subsubsection*{Step 1: Fast Elimination}
To prove that an edge $pq$ is useless we need 
to find two potential points $r$ and $s$ satisfying the conditions of the 
Main Edge Elimination Theorem. For a point $r$ we use the method described in 
Section~\ref{sec:certification}, to prove that it is potential. In fact we will 
only use $r$ if it is strongly potential. This can be checked in constant time.
The next step is to calculate the minima appearing in the Main Edge Elimination Theorem
using Lemma~\ref{lemma:min-estimation}.
This can be done separately for each potential point in constant time.

Once two potential points $r$ and $s$ with their corresponding minima are 
calculated, one can check in constant time whether the inequalities of the Main 
Edge Elimination Theorem are satisfied. 
Since only two potential points which satisfy the conditions of the theorem are 
needed, a smart ordering and stopping criteron for checking the
potentiality of points can speed up the algorithm drastically.
We select the points ordered by their distance from the midpoint of the edge $pq$
and stop after at most $10$ points that have been considered.

Using a 2-d tree~\cite{Ben1975} we compute in a preprocessing step 
the values $\delta_r$ for all vertices of the instance. 
In most cases it turns out that at most three strongly potential points have to be considered
to prove that an edge is useless.

\subsubsection*{Step 2: Direct Elimination}
For an edge $pq$ we consider two vertices $r$ and $s$ with all their 
incident edge pairs. The Close Point Elimination is used to 
reduce the number of possible edge pairs.
The edge $pq$ can be eliminated either if no edge pair can be found for $r$ or 
$s$, or if all combinations of edge pairs satisfy the condition of the 
Main Edge Elimination Theorem.

\subsubsection{Step 3: Backtrack Search}
\label{sec:backtracking}
In this step starting with an edge $pq$ we extend a set of disjoint paths recursively.
We allow two operations for the extension. Either we select one of the existing paths and add
an edge incident to one of its endpoints. Or we add a vertex not yet contained in any of
the disjoint paths and two edges incident with this vertex. 
After each extension we check whether the Main Edge Elimination Theorem or the Close Point Elimination 
allows to eliminate one of the path edges. In this case we backtrack. 
Moreover we check that the collection of paths is minimal in the sense
that no collection of paths exists that has shorter length and that connects the same pairs of 
endpoints and uses the same set of interior points. We use an extension of the 
Held-Karp algorithm~\cite{HK1962} for this.  
We always select the extension of the set of paths that has the smallest number of possibilities. 
If all extensions have been examined
without reaching a predefined extension depth, then we have proven that edge $pq$ is useless.

\section{Results}
\label{sec:results}

We applied our algorithm to all TSPLIB~\cite{Rei1995} instances which 
use the EUC\_2D metric as well as to some larger EUC\_2D instances from~\cite{BCWebsite}.
Table~\ref{tab:results} contains the results on some of these instances.

\begin{table}[ht]
\caption{Results for some TSPLIB instances with at least 1,000 vertices as well as a 100,000 vertex instance from~\cite{BCWebsite}.
The first three columns contain the instance name, the number of vertices and the number of edges. 
Then for each of the three steps as described in Section~\ref{sec:algorithm} we list the number of edges
that remain after this step as well as the runtime. For Step~3 we used an extension depth of~10 which gave a reasonable trade off between the runtime and 
the number of eliminated edges. 
The last two columns contain the 
total runtime of our algorithm and the ratio of the number of edges remaining after Step~3 divided by the number of vertices. 
All runtimes are given in the format \texttt{hh:mm:ss} and are measured
on a single core of a 2.9GHz Intel Xeon.}
\label{tab:results}
\begin{center}
\tiny
\scriptsize
\begin{tabular}{|l|r|r||r|r||r|r||r|r||r|r|}
\hline
 &  &  &  \multicolumn{2}{c||}{~Step 1}  & \multicolumn{2}{c||}{~Step 2} & \multicolumn{2}{c||}{~Step 3} &  total &  \\
\raisebox{1.5ex}[-1.5ex]{\textbf{instance}} & \raisebox{1.5ex}[-1.5ex]{$n$} & \raisebox{1.5ex}[-1.5ex]{$m$} &     edges &   time &   edges &  time &  edges &  time &
 runtime &\raisebox{1.5ex}[-1.5ex]{ratio}\\
\hline
pr1002            &       1002 &     501501 &      42636 &         1 &       5810 &       2:13 &    4521  &   2:28:07  &    2:30:21 &  4.2   \\ \hline
u1060             &       1060 &     561270 &      43887 &         1 &       6063 &       2:24 &    4619  &   3:30:48  &    3:33:13 &  4.4   \\ \hline
vm1084            &       1084 &     586986 &      40958 &         1 &       6035 &       3:28 &    4610  &   1:17:35  &    1:21:05 &  4.3   \\ \hline
pcb1173           &       1173 &     687378 &      32533 &         1 &       7662 &         33 &    6084  &   3:10:17  &    3:10:51 &  5.2   \\ \hline
d1291             &       1291 &     832695 &     122897 &         4 &      12552 &      52:21 &   11317  &  13:33:14  &   14:25:40 &  8.8   \\ \hline
rl1304            &       1304 &     849556 &     124561 &         3 &      21689 &      11:54 &   14527  &  12:53:14  &   13:05:11 & 11.1   \\ \hline
rl1323            &       1323 &     874503 &     106860 &         2 &      16743 &       5:33 &   12691  &   9:41:18  &    9:46:53 &  9.6   \\ \hline
nrw1379           &       1379 &     950131 &      28468 &         1 &       7199 &       1:58 &    5752  &   2:12:44  &    2:14:44 &  4.2   \\ \hline
u1432             &       1432 &    1024596 &      21970 &         1 &       7817 &       2:51 &    6495  &   2:13:02  &    2:15:55 &  4.5   \\ \hline
d1655             &       1655 &    1368685 &     230855 &         9 &      14345 &      37:30 &   12103  &  10:05:46  &   10:43:26 &  7.3   \\ \hline
vm1748            &       1748 &    1526878 &     144681 &         6 &      12303 &      23:10 &    7691  &   2:48:14  &    3:11:30 &  4.4   \\ \hline
u1817             &       1817 &    1649836 &     109056 &         5 &      13201 &      10:58 &   11736  &   6:19:22  &    6:30:25 &  6.5   \\ \hline
rl1889            &       1889 &    1783216 &     206768 &         9 &      23410 &    3:15:40 &   18673  &  25:18:52  &   28:34:41 &  9.9   \\ \hline
d2103             &       2103 &    2210253 &     166866 &         8 &      19631 &      55:01 &   18105  &  18:19:34  &   19:14:44 &  8.6   \\ \hline
u2152             &       2152 &    2314476 &     117030 &         5 &      15101 &      11:17 &   13170  &   7:07:45  &    7:19:08 &  6.1   \\ \hline
u2319             &       2319 &    2687721 &      21698 &         3 &       9919 &         44 &    9473  &   1:41:41  &    1:42:22 &  4.1   \\ \hline
pr2392            &       2392 &    2859636 &     121514 &         7 &      15598 &      13:03 &   12088  &   7:41:45  &    7:44:55 &  5.1   \\ \hline
pcb3038           &       3038 &    4613203 &      95576 &         8 &      17940 &      11:05 &   14869  &   5:44:08  &    5:55:22 &  4.9   \\ \hline
fnl4461           &       4461 &    9948030 &     128527 &        15 &      23963 &       9:30 &   19082  &   7:14:21  &    7:24:07 &  4.3   \\ \hline
brd14051          &      14051 &   98708275 &    2661869 &      4:39 &      93497 &   18:50:15 &   64486  &  28:39:22  &   47:34:16 &  4.6   \\ \hline
d15112            &      15112 &  114178716 &    1703765 &      5:51 &     130110 &   10:25:44 &   66010  &  38:37:42  &   49:09:17 &  4.4   \\ \hline
d18512            &      18512 &  171337816 &    1449877 &      5:30 &     112681 &    1:49:35 &   84203  &  32:38:07  &   34:33:13 &  4.5   \\ \hline
mona-lisa100k     &     100000 & 4999950000 &    2071297 &   3:45:21 &     476001 &      22:51 &  322716  &  55:42:51  &   59:51:04 &  3.2   \\ \hline
\end{tabular}\vspace*{-12mm}
\end{center}
\end{table}

The TSP solver Concorde~\cite{ABCC2006} is the fastest available algorithm 
to solve large TSP instances optimally. It can be downloaded at~\cite{BCWebsite}. 
We applied Concorde to the TSPLIB instance d2103. 
The total runtime needed by Concorde was $17,219,190$~seconds\footnote{Runtime on a 2.9GHz Intel Xeon using Concorde~\cite{BCWebsite}. 
We observed that the runtime of Concorde can vary by more than 40\% on the same instance.
Therefore, we took the average runtime of two independent runs. The Concorde log-files for all runs can be found at 
\texttt{http://www.or.uni-bonn.de/\textasciitilde hougardy/EdgeElimination}.}.
This agrees with the runtime reported for this instance 
on page 503 of~\cite{ABCC2006}. Then we ran the three steps of our algorithm as described in Section~\ref{sec:algorithm}.
For Step~3 we used an extension depth of~12. After $168,153$~seconds all but $16,566$~edges were eliminated. 
We changed the length of all eliminated edges
to some large value and gave this new instance again to Concorde. On this instance Concorde needed $1,392,582$~seconds.
Thus the total runtime was improved by our edge elimination algorithm by more than a factor of 11.

Two other successful applications of our edge elimination approach were reported to us by Cook~\cite{BCpersonal}. 
First, the edge elimination approach in combination with the LKH algorithm~\cite{Hel2009} improved the so far best known TSP tour for
the DIMACS instance E100k.0~\cite{KHWebpage}. The shortest tour known so far had length $225, 786, 982$.
It was found using the LKH algorithm.
Cook's implementation of an edge elimination approach eliminated all but $274,741$ edges in the TSP instance E100k.0.
By applying the LKH algorithm to this edge set a tour of length $225,784,127$ was found~\cite{KHWebpage}.
Secondly, Cook applied the edge elimination approach to a truly Euclidean instance 
(i.e., a Euclidean instance where the point distances are not rounded). 
Finding optimum TSP tours in such instances is much harder than in instances with rounded Euclidean norm. 
The largest truly Euclidean instance that Cook was able to solve so far had 500 points. With the help of the edge elimination approach he
solved an instance with 1000 points.

%
%


\newpage
\section{Appendix}
\label{sec:appendix}

\subsection{The angles $\alpha_p$ and $\alpha_q$}

The following lemma shows how to compute the angles  $\alpha_p$ and $\alpha_q$,
discussed in Section~\ref{sec:alphas}.

\begin{lemma}
Let $(V,E)$ be TSP instance, $pq\in E$ and $r\in V\setminus \{p,q\}$.
Then the angles $\alpha_p$ and $\alpha_q$ of the cones $R_p$ respectively
$R_q$ satisfy
 \begin{eqnarray}
  \alpha_p &=& 
  2\cdot \arccos\left(\frac{l_q^2 - \delta_r^2 - |rq|^2}{2\delta_r |rq|}\right)
  \label{eqn:alpha_p}~~~~\text{respectively}\\
  \alpha_q &=& 
  2\cdot \arccos\left(\frac{l_p^2 - \delta_r^2 - |rp|^2}{2\delta_r |rp|}\right)~.
  \label{eqn:alpha_q}
 \end{eqnarray}
\end{lemma}

\begin{proof}
Let $C_r$ be the circle around $r$ with radius $\delta_r$ and $C_q$ the circle
around $q$ with radius $l_q$. Let $t$ be one intersection point of these two
circles. Consider the triangle with the points $q$, $r$ and $t$. Let $\omega_q$ 
be the angle of the triangle at the point $r$. 
The cosine theorem then yields
\begin{equation}
 \delta_r^2 + |rq|^2 - 2\delta_r|rq|\cos(\omega_q) = l_q^2~. \label{eqn:cos-omega}
\end{equation}
Hence we can conclude
\begin{eqnarray}
 \cos\left(\frac{\alpha_p}{2}\right) 
 &=& \cos\left(180^\circ - \omega_q\right)\\
 &=& -\cos(\omega_q)\\
 &\overset{(\ref{eqn:cos-omega})}{=}&  
 \frac{l_q^2 - \delta_r^2 - |rq|^2 }{2\delta_r|rq|}~.
\end{eqnarray}
This equation implies~(\ref{eqn:alpha_p}).
Equation~(\ref{eqn:alpha_q}) can be shown equivalently.
\end{proof}

\subsection{Validating}

In order to calculate the maxima of (\ref{eqn:MinMax-p}) and (\ref{eqn:MinMax-q}), 
we need to define some angles. For an edge $pq\in E$ and a point 
$r\in V\setminus \{p,q\}$ we define $\varepsilon_p$ and $\varepsilon_q$ as the 
angles in the triangle $pqr$ at the points $p$ respectively $q$. Further let $t$
be such that $|rt|=\delta_r$ and $|pt|=l_p$. Let $\theta_p$ be the angle in the 
triangle $prt$ at the point $p$. The angle $\theta_q$ is defined 
equivalently. By the cosine theorem we obtain the following equations:

\begin{figure}[ht]
\begin{center}
\begin{tikzpicture}

\coordinate (p) at (0,0);
\coordinate (q) at (6,0);
\coordinate (r) at (3,2);
\filldraw (p) node[anchor=east]{$p$} circle (2pt);
\filldraw (q) node[anchor=west]{$q$} circle (2pt);
\filldraw (r) node[anchor=north]{$r$} circle (2pt);
\draw (p) -- (q);

\draw (r) circle (1.2);

\draw (r) -- node[anchor=west]{$\delta$} +(100:1.2) -- (q);
\draw (r) -- +(190:1.2) -- (q);

\draw (q) -- +(147:0.6) arc (147:180:0.6);
\draw (q) -- +(147:1.3) arc (147:157:1.3);
\draw (5.2,0.15) node{$\varepsilon_q$};
\draw (4.6,0.8) node{$\theta_q$};

\begin{scope}
\clip (0,1.2) rectangle (6,3.5);
\draw (q) circle (4.5);
\end{scope}

\filldraw (p) -- (r);
\filldraw (q) -- (r);
\end{tikzpicture}
\caption{The angles $\varepsilon$ and $\theta$}
\label{fig:eps-theta}
\end{center}
\end{figure}
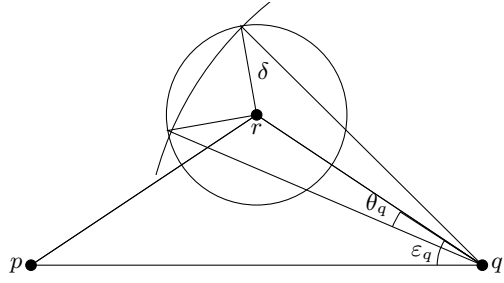

\begin{eqnarray}
\cos\varepsilon_p &=& \frac{|pq|^2 + |pr|^2 - |qr|^2}{2|pq||pr|}\\
\cos\varepsilon_q &=& \frac{|pq|^2 + |qr|^2 - |pr|^2}{2|pq||qr|}\\
\cos\theta_p &=& \frac{l_p^2 + |pr|^2 - \delta_r^2}{2l_p|pr|}\\
\cos\theta_q &=& \frac{l_q^2 + |qr|^2 - \delta_r^2}{2l_q|qr|}~.
\end{eqnarray}

By considering the cosine theorem in the triangles $pq\tilde{p}$ 
respectively $pq\tilde{q}$, the inequalities~(\ref{eqn:pq-tilde}) 
are equivalent to
\begin{eqnarray}
\frac{(|qr| + \delta_r)^2 + |pq|^2 - l_p^2}{2(|qr| + \delta_r)|pq|} 
&\leq& \cos\varepsilon_q\\
\frac{(|pr| + \delta_r)^2 + |pq|^2 - l_q^2}{2(|pr| + \delta_r)|pq|} 
&\leq& \cos\varepsilon_p~.
\end{eqnarray}

We want to calculate the maxima of Lemma~\ref{lemma:min-estimation}. They 
correspond to
\begin{eqnarray}
 \text{max}\{|ps_r|: s\in R_p\} &=& \text{max}\{b_p, b_p'\}\\
 \text{max}\{|qs_r|: s\in R_q\} &=& \text{max}\{b_q, b_p'\}~,
\end{eqnarray}
whereas the endpoints of the archs $B_p$ and $B_q$ are defined as $b_p$ and
$b_p'$ respectively $b_q$ and $b_q'$.

\begin{lemma}
Assuming $\varepsilon \leq 180^\circ$ and $\theta \leq 180^\circ$ then we have
the following equivalence:
\begin{equation}
\varepsilon + \theta \leq 180^\circ \Leftrightarrow \cos\varepsilon +
\cos\theta \geq 0
\end{equation}
\end{lemma}

Using simple geometric calculations, we can deduce the following lemma:

\begin{lemma}
\label{lemma:estimation}
Let $r$ be strongly potential with respect to $pq$, and let the 
inequalities~(\ref{eqn:pq-tilde}) hold. 
Let $\varepsilon_p + \theta_p \leq 180^\circ$ and
$\varepsilon_q + \theta_q \leq 180^\circ$. Then
\begin{eqnarray}
\left(\text{max}\{|ps_r|: s\in R_p\}\right)^2
&=& |pq|^2 + l_q^2 - 2|pq|l_q\cos(\varepsilon_q + \theta_q)\\
\left(\text{max}\{|qs_r|: s\in R_q\}\right)^2
&=& |pq|^2 + l_p^2 - 2|pq|l_p\cos(\varepsilon_p + \theta_p)
\end{eqnarray}
\end{lemma}

\begin{proof}
From the assumption we can follow $\cos(\varepsilon_q - \theta_q) \geq
\cos(\varepsilon_q + \theta_q)$. The cosine theorem then implies
\begin{equation*}
\left(\text{max}\{|ps_r|: s\in R_p\}\right)^2
= |pq|^2 + l_q^2 - 2|pq|l_q\cos(\varepsilon_q + \theta_q)~.
\end{equation*}
\qed
\end{proof}


\begin{thebibliography}{5}
%

\bibitem
   {ABCC2006}
   {D.L.Applegate, R.E.Bixby, V.Chv\'atal, W.J.Cook.}
   {The Traveling Salesman Problem. A Computational Study.}
   {Princeton University Press, 2006}

\bibitem
   {Ben1975}
   {J.L.Bently.}
   {Multidimensional binary search trees used for associative searching.}
   {Communications of the ACM, Volume 18, Number 9 (September 1975), pp. 509--517}

\bibitem
   {BCWebsite}
   {W.J.Cook's TSP website at \texttt{http://www.math.uwaterloo.ca/tsp/}}
   {}
   {}

\bibitem
   {BCpersonal}
   {W.J.Cook, personal communication, December 2013.}
   {}
   {}

\bibitem
   {GarJoh1979}
   {M.R.Garey, D.S.Johnson.}
   {Computers and Intractability.}
   {Freeman, New York 1979}
   
\bibitem{HK1962}
   {M.Held, R.M.Karp} 
   {A Dynamic Programming Approach to Sequencing Problems.}
   {Journal of the Society for Industrial and Applied Mathematics, Vol. 10, No. 1 (Mar.,1962), pp. 196-210}
   
\bibitem{Hel2009}
   {K.Helsgaun.}
   {General $k$-opt submoves for the Lin-–Kernighan TSP heuristic.}
   {Mathematical Programming Computation, Volume 1, Issue 2-3 (2009), pp. 119--163}   


\bibitem{KHWebpage}
   {K.Helsgaun,}
   {http://www.akira.ruc.dk/\textasciitilde keld/research/LKH/DIMACS\_results.html}
   {}


\bibitem
   {Rei1995}
   {G.Reinelt.}
   {TSPLIB 95.}
   {Interdisziplin\"ares Zentrum f\"ur Wissenschaftliches Rechnen (IWR), Heidelberg, 1995}

\end{thebibliography}
\end{document}